\documentclass[envcountsame,runningheads]{llncs}

\usepackage[T1]{fontenc}
\usepackage[utf8]{inputenc}
\usepackage{lmodern}
\usepackage{microtype}
\input{glyphtounicode}
\usepackage{amsmath}
\usepackage{amsfonts}
\usepackage{amssymb}

\usepackage{mathtools}
\mathtoolsset{mathic=true}

\usepackage[ruled,vlined,linesnumbered]{algorithm2e}
\SetKw{KwDownTo}{downto}
\SetKwComment{Comment}{$\triangleright$~}{}
\SetCommentSty{itshape}
\DontPrintSemicolon

\usepackage{lineno}
\modulolinenumbers[5]

\usepackage{graphicx}
\usepackage{tabto}
\usepackage{xspace}
\usepackage[caption=false]{subfig}

\usepackage[inline,shortlabels]{enumitem}
\setlist[enumerate]{leftmargin=*}

\DeclareRobustCommand{\qedClaim}
{%
  \ifmmode \lozenge%
  \else%
    \leavevmode\unskip\penalty9999 \hbox{}\nobreak\hfill%
    \quad\hbox{$\lozenge$}%
  \fi%
}%

\DeclareRobustCommand{\qed}
{%
  \ifmmode \square%
  \else%
    \leavevmode\unskip\penalty9999 \hbox{}\nobreak\hfill%
    \quad\hbox{$\square$}%
  \fi%
}%

\spnewtheorem*{conj}{Conjecture}{\normalfont\bfseries}{\itshape}
\spnewtheorem{obs}{Observation}{\rmfamily}{\itshape}
\spnewtheorem{subclaim}{Claim}{\rmfamily}{\itshape}

\newcommand{\calB}{\mathcal{B}}

\newcommand{\calE}{\mathcal{E}}
\newcommand{\calF}{\mathcal{F}}

\newcommand{\calO}{\mathcal{O}}

\newcommand{\calS}{\mathcal{S}}

\makeatletter
\newcommand{\ie}{i.\,e.\@ifnextchar{,}{}{~}}
\newcommand{\eg}{e.\,g.\@ifnextchar{,}{}{~}}
\newcommand{\vtx}[1]{\textsl{#1}\@ifnextchar{)}{\hspace{0.08333em}}{}}
\makeatother

\title{%
    Canonical Join Trees%
}

\author{%
    Arne Leitert
}

\institute{}

\DeclareMathOperator{\dep}{dep}
\DeclareMathOperator{\dis}{d}

\DeclareMathOperator{\edge}{\text{---}}
\DeclareMathOperator{\path}{\tikz[baseline=-0.625ex]{\draw [line width =0.25pt,decorate, decoration={zigzag, segment length=0.5em, amplitude=0.2ex}] (0,0) -- (1.em,0);}}

\usepackage{tikz}
\usetikzlibrary{decorations.pathmorphing}

\begin{document}

\maketitle

% \linenumbers

\begin{abstract}
A rooted join tree of an acyclic hypergraph is \emph{canonical} if each of its nodes has minimum possible depth among all join trees with the same root.
Luo~et\,al.\  introduce these trees in~\cite{LuoVavVanWan2026} and pose the open problem of characterizing acyclic hypergraphs according to whether they admit canonical join trees for none, some, or all hyperedges as root.
In this paper, we resolve this question.
We show that each canonical join tree is unique with respect to its root and give a first characterisation for such trees.
Additionally, we characterise hypergraphs that admit a canonical join tree for none, some, or all their hyperedges as root.
Lastly, we present a linear-time algorithm that constructs a canonical join tree whenever one exists.
\end{abstract}

\section{Introduction}

Acyclic hypergraphs are a fundamental structure in database theory and graph theory.
Their key characterisation is the ability to arrange their hyperedges into a tree such that, if two hyperedges share a vertex~$\vtx{x}$, then every hyperedge on the path between them also contains~$\vtx{x}$.
Such a tree is called \emph{join tree}.
They provide a compact structural representation of acyclic hypergraphs and form the basis of many algorithms on them.
Recently, join trees have also attracted renewed attention in the database community as tool for optimizing join queries (see for example \cite{KouVanWanWan2026,LuoVavVanWan2026,WangYi2026} and literature cited in them).

Out of all join trees an acyclic hypergraph has, trees with small height are of special interest, because such trees allow better parallelization and index utilization.
In~\cite{BlairPeyton1994}, Blair and Peyton present an algorithm which computes a join tree with minimum diameter (and thus minimum height) in $\calO(N)$ time where $N$ is the total input size of a given hypergraph.
Note that their algorithm assumes the maximal cliques of a chordal graph as input.
One can therefore use it for any acyclic hypergraph without the overhead of creating a corresponding chordal graph.
With the same motivation, Luo~et\,al.\,\cite{LuoVavVanWan2026} introduce the stricter notion of a \emph{canonical} join tree.
Such a tree~$T$ simultaneously minimises the depth of every individual node~$u$ (denoted as~$\dep_T(u)$) instead of just the overall height:

\begin{definition}\label{def:canJoinTree}
A join tree~\( T \) with root~\( r \) is \emph{canonical} if \( \dep_{T}(u) \leq \dep_{T'}(u) \) for any other join tree~\( T' \) rooted in~\( r \) and all nodes~\( u \).
\end{definition}

Luo~et\,al.\ demonstrate that Berge-acyclic hypergraphs possess a unique canonical join tree for each of their hyperedges as root, and that one can find this tree in linear time.
They also raise the question of precisely characterising hypergraphs which admit unique canonical join trees for none, some, or all their hyperedges as root, respectively.
We answer this question in this paper.
Our specific contributions are as follows:
\begin{itemize}
    \item We show that, every canonical join tree is unique if it exists.
    \item We characterize canonical join trees in terms of their hypergraph's union join graph.
    \item We provide a subgraph characterization for acyclic hypergraphs that do not admit any canonical join trees.
    \item We present a linear-time algorithm which finds a canonical join tree for a given hypergraph and root, provided one exists.
    \item We characterise hypergraphs that admit a canonical join tree for every hyperedge as root.
\end{itemize}

\section{Preliminaries}

Let $H = (V, \calE)$ be a hypergraph.
We use $N = \sum_{E \in \calE} |E|$ to denote the total size of all hyperedges of~$H$.
Whenever a hypergraph is given, the input size is in~$\Theta(N)$.

This paper primarily investigates graphs and trees constructed from the hyperedges of some hypergraph.
We therefore use the term \emph{vertex} exclusively for elements with in a hyperedge, and we use the term \emph{node} for hyperedges in context of such a graph.
That is, a \emph{node} and \emph{hyperedge} are interchangeable terms in this paper.
We use slanted lowercase letters (\eg $\vtx{x}$, $\vtx{y}$) to represent vertices and italic lowercase letters (\eg $u$, $v$) to identify nodes.
Keep in mind, however, that nodes are sets of vertices.
Thus, $u = v$ indicates that both nodes contain the same vertices.
To indicate that both variables reference the same individual node, we write $u \equiv v$ instead.
Similarly, we also treat an edge~$uv$ as the intersection of both nodes, \ie, $uv = u \cap v$.
We also treat a graphs and trees as sets of nodes and edges.
For example, we write $uv \in G$ to indicate that $uv$ is an edge of~$G$, or write $x \notin T$ to indicate that $x$ is not a node of~$T$.

In a graph~$G$, the \emph{distance~\( \dis_G(u, v) \)} of two nodes $u$ and~$v$ is the length of a shortest path connecting $u$ and~$v$.
The \emph{depth~\( \dep_T(v) \)} of a node~$v$ in a rooted tree~$T$ is the distance from $v$ to~$T$'s root~$r$, \ie, $\dep_T(v) = \dis_T(v, r)$.
If clear from context, we may omit graph identifying indices of notations.

For a tree~$T$, $T[u, v]$ denotes the path from~$u$ to~$v$ in~$T$.
As with graphs and trees, we treat $T[u, v]$ as a set of nodes and edges.
We write $u \edge v$ to denote that $u$ is adjacent to~$v$ and $u \path v$ to denote that there is a path of from $u$ to~$v$.
The latter does not specify the length of that path and includes the case that $u \equiv v$.
If $T$ is a rooted spanning tree of some graph~$G$, we say an edge~$uv \in G$ is a \emph{tree edge} if $uv \in T$, $uv$ is a \emph{back edge} if $uv \notin T$ and $u$ is an ancestor of~$v$ in~$T$ or vice versa, and $uv$ is a \emph{cross edge} if it is neither a tree nor a back edge.

A tree~$T$ is called a \emph{join tree} for~$H$ if the hyperedges of~$H$ are the nodes of~$T$ and, for each vertex~$\vtx{x}$, the hyperedges containing~$x$ induce a subtree of~$T$.
That is, for any pair of nodes~$u, v \in T$, if $\vtx{x} \in u \cap v$, then $\vtx{x}$ is contained in each node in~$T[u, v]$.
A hypergraph is called \emph{acyclic} if it admits a join tree.
It is well known that one can determine if a given hypergraph is acyclic and, in that case, constructs a corresponding join tree for it in linear time~\cite{TarjanYannak1984}.
See~\cite{BrandsDragan2015} for a summary of known properties of acyclic hypergraphs as well as their relations to various graph classes.

The \emph{line graph~\( L \)} of~$H$ is the intersection graph of its hyperedges.
That is, $L$ has $H$'s hyperedges as nodes and $uv \in L$ if $u \cap v \neq \emptyset$.
It is well known that a tree~$T$ is a join tree for~$H$ if and only if it is the maximum spanning tree of $H$'s line graph where the weight of an edge is the number of vertices both hyperedges share~\cite{GaliHabiPaul1995,MayhewProber2024}.

Consider a rooted join tree~$T$ and a node~$u \in T$ with parent~$p$.
Then $q$ is the \emph{highest possible ancestor} of~$u$ if $q$ is the ancestor of~$u$ closest to $T$'s root such that $q \supseteq up$.
Accordingly, \emph{making all nodes adjacent to their highest possible ancestor} is the process of finding such $q$ for each node~$u$, removing the edge~$up$ and adding the edge~$uq$ instead.
This can be done in linear time for any given join tree~\cite{Leitert2021}, or one can directly compute such a join tree from a given acyclic hypergraph in linear time~\cite{LuoVavVanWan2026}.

\section{Properties of Join Trees and Their Unions}
\label{sec:joinTreeEdges}

A given acyclic hypergraph~$H$ may have up to exponentially many join trees.
The union of these trees is called a \emph{union join graph}.
That is, a graph~$G$ is a union join graph if its nodes are the hyperedges of~$H$ and two nodes $u$ and~$v$ are adjacent in~$G$ if $H$ admits a join tree~$T$ with an edge~$uv$.

These graphs are also known as \emph{reduced clique graphs} if $H$ represents the maximal cliques of a chordal graph~\cite{GaliHabiPaul1995,HabibStacho2012}, or \emph{atom graphs} if $H$ represents the atoms of some graph~\cite{KabPinLelBer2007}.
Under reasonable assumptions, it is not possible to compute the union join graph of a given acyclic hypergraph is $\calO\bigl( N^{2 - \varepsilon} \bigr)$ time for any constant~$\varepsilon > 0$, although faster algorithms are available for specific subclasses~\cite{Leitert2021}.

% Our goal in this section is to derive structural properties of union join graphs that later allow us to characterize canonical join trees.

In this paper, we use the relations between join trees and union join graphs as one of the main tools to characterise canonical join trees.
We therefore use this section to prove some useful properties.
In what follows, assume that we are given an acyclic hypergraph~$H$ and that $G$ is its union join graph.
If not specified or constructed otherwise, we denote with~$T$ an arbitrary join tree of~$H$.

Lemma~\ref{lem:unionJoinProperties} and Lemma~\ref{lem:bipartiteClique} below describe properties between edges and necessary adjacencies between nodes in $G$ and~$T$.

\begin{lemma}[\cite{GaliHabiPaul1995,HabibStacho2012}]
\label{lem:unionJoinProperties}
For any distinct nodes \( x, y \), the following are equivalent.
\begin{enumerate}[$(i)$]
    \item
        \label{item:unionJoinGrpahEdge}
        \( xy \in G \).
    \item
        \label{item:joinTreeEdge}
        \( H \) admits a join tree~\( T \) with \( xy \in T \).
    \item
        \label{item:intersectionOnPath}
        Each join tree~\( T \) of~\( H \) has an edge~\( uv \in T[x, y] \) such that \( uv = xy \).
\end{enumerate}
\end{lemma}

\begin{lemma}\label{lem:bipartiteClique}
If \( xy \in G \), then \( T[x, y] = x \path u \edge v \path y \) with \( uv = xy \).
Additionally, \( pq \in G \) for all nodes~\( p \in T[x, u] \) and all nodes~\( q \in T[v, y] \).
\end{lemma}

\begin{proof}
It follows from Lemma~\ref{lem:unionJoinProperties} that $T[x, y]$ contains an edge~$uv = xy$.
For all nodes~$p \in T[x, u]$ and all nodes~$q \in T[v, y]$, we can observe that $xy \subseteq p \cap q$  since $p, q \in T[x, y]$, and that $p \cap q \subseteq uv$ since $uv \in T[p, q]$.
In other words, $xy = pq = uv$.
We can therefore create a new join tree from~$T$ by removing~$uv$ and replacing it with $pq$ instead.
\qed
\end{proof}

Consider a non-tree edge~$xy \in G$ with respect to some join tree~$T$.
It follows from Lemma~\ref{lem:bipartiteClique} that no matter how far apart $x$ and~$y$ are in~$T$, $G$ also contains a non-tree edge~$uw$ with a shared neighbour~$v$ in~$T$.
Corollary~\ref{cor:2LayerBackEdge} formalises that observation.

\begin{corollary}\label{cor:2LayerBackEdge}
If \( xy \in G \) is a non-tree edge with respect to some join tree~\( T \), then \( T[x, y] \) contains a path~\( u \edge v \edge w \) such that \( uw \in G \).
Additionally, either \( uw = uv \subseteq vw \) or \( uv \supseteq vw = uw \) (\ie, \( uw \) is equal to the ``smaller'' of \( uv \) and~\( vw \)).
\end{corollary}

If an edge~$uw$ as defined in Corollary~\ref{cor:2LayerBackEdge} is a back edge, then we call it a \emph{2-layer back edge}.
Algorithm~\ref{algo:All2LayerBackEdge} allows us to find all such edges for a given join tree in linear time.
We define a \emph{top-down order} for a rooted tree as any order which visits parents before their children (\eg a pre-order or a BFS-order).

\SetKwFor{TDOrder}{foreach}{in top-down order}{}

\begin{algorithm}
\caption{%
    Finds all 2-layer back edges induced by a given join tree.
}
\label{algo:All2LayerBackEdge}

\KwIn{%
    An acyclic hypergraph~$H$ with a rooted join tree~$T$.
}

\KwOut{%
    All tripples $u,v,w$ which form a 2-layer back edge with respect to~$T$.
}

\medskip

Set $\phi(\vtx{x}) := \infty$ for each vertex~$\vtx{x}$ of~$H$.

\TDOrder{%
    node~\( w \in T \)%
    \label{line:iterNodes}
}{%
    \smallskip

    Set $\alpha(w) := 0$ and $\beta(w) := 0$.

    \ForEach{%
        vertex~\( \vtx{x} \in w \)%
        \label{line:iterVertices}
    }{%
        \lIf{%
            \( \phi(\vtx{x}) = \infty \)
        }{%
            Set $\phi(\vtx{x}) := \dep_T(w)$.
        }

        \lIf{%
            \( \phi(\vtx{x}) < \dep_T(w)\)
        }{%
            Set $\alpha(w) := \alpha(w) + 1$.
        }

        \lIf{%
            \( \phi(\vtx{x}) < \dep_T(w) - 1 \)
        }{%
            Set $\beta(w) := \beta(w) + 1$.
        }
    }

    \smallskip

    Let $v$ be the parent and $u$ be the grand parent of~$w$ (if they exist).

    \If{%
        \( \dep_T(w) \geq 2 \) and \( \beta(w) \in \bigl \{ \alpha(v), \alpha(w) \bigr \} \)
    }{%
        \KwSty{print} $u, v, w$
    }
}
\end{algorithm}

\begin{lemma}\label{lem:All2LayerBackEdge}
Algorithm~\ref{algo:All2LayerBackEdge} computes all 2-layer back edges induced by a given join tree in $\calO(N)$ time.
\end{lemma}

\begin{proof}
We first observe that $\phi(\cdot)$ states for a vertex the depth of the highest node containing it.
The values $\alpha(\cdot)$ and~$\beta(\cdot)$ then count the number of vertices a node shares with its parent and grand parent, respectively.
Therefore, if $u \edge v \edge w$ is a path in~$T$ such that $u$ is the parent of~$v$ and $v$ is the parent of~$w$, then $\alpha(v) = |uv|$, $\alpha(w) = |vw|$, and $\beta(w) = |uv \cap vw|$.
To prove the correctness of Algorithm~\ref{algo:All2LayerBackEdge}, we show that $uw \in G$ if and only if $\alpha(v) = \beta(w)$ or $\alpha(w) = \beta(w)$.

First, assume that $uw \in G$.
By Corollary~\ref{cor:2LayerBackEdge}, either $uw = uv \subseteq vw$, or $uv \supseteq vw = uw$.
In the first case, $uv \subseteq vw$ implies that $\beta(w) = |uv \cap vw| = |uv| = \alpha(v)$.
In the second case, $uv \supseteq vw$ implies that $\beta(w) = |uv \cap vw| = |vw| = \alpha(w)$.

Next, assume that $\alpha(v) = \beta(w)$ or $\alpha(w) = \beta(w)$.
It follows from the properties of join trees and the definitions of $\alpha$ and~$\beta$ that $u \cap w = uv$ or $u \cap w = vw$.
In either case, Lemma~\ref{lem:unionJoinProperties} then implies that $uw \in G$.

To prove the algorithm's complexity, observe that it iterates over each hyperedge~$w$ of~$H$ (line~\ref{line:iterNodes}) and then over each vertex~$\vtx{x}$ contained in~$w$ (line~\ref{line:iterVertices}).
The algorithm then performs a few constant-time operations for each such iteration.
Therefore, Algorithm~\ref{algo:All2LayerBackEdge} runs in $\calO(N)$ total time.
\qed
\end{proof}

\section{Canonical Join Trees}

In this section, we present the main results of our paper.
We start by characterising canonical join trees themself and later characterise acyclic hypergraphs based on whether or not they admit such trees.

\begin{lemma}\label{lem:canUnique}
A canonical join tree is always unique with respect to its root.
\end{lemma}

\begin{proof}
Let us assume that there are two distinct canonical join trees $T_1$ and~$T_2$ with the same root.
Note that, by definition, $\dep_{T_1}(u) = \dep_{T_2}(u)$ for each node~$u$.
Therefore, some nodes must have different parents in $T_1$ and~$T_2$.
Let $u$ be a highest such node with parents $p_1$ and~$p_2$ in $T_1$ and~$T_2$, respectively.
That is, all nodes~$v$ with $\dep(v) < \dep(u)$ have the same parent in both trees.

Recall that, by properties of join trees, $up_2 \subseteq v$ for all nodes~$v \in T_1[u, p_2]$.
That includes $p_1$ as well as the lowest common ancestor~$q$ of $p_1$ and~$p_2$.
Thus, since $up_2$ is an edge of~$T_2$, we can create a join tree~$T_3$ by removing $up_2$ from~$T_2$ and adding edge~$uq$ instead.
However, since $\dep(q) < \dep(p_1) = \dep(p_2)$, $\dep_{T_3}(u) < \dep_{T_2}(u)$.
This contradicts with our initial assumption that $T_1$ and~$T_2$ are canonical.
\qed
\end{proof}

\begin{theorem}\label{theo:canonIffBackEdge}
A rooted join tree~\( T \) is canonical if and only if \( T \) does not induce any back edges in its union join graph.
\end{theorem}

\begin{proof}
$\rightarrow$:
Assume that $T$ induces a back edge~$uv$ with $u$ being an ancestor of~$v$, \ie, $\dep_T(u) + 1 < \dep_T(v)$.
By Lemma~\ref{lem:unionJoinProperties}, there is an edge~$xy \in T[u, v]$ with $uv = xy$.
We can therefore create a new join tree~$T'$ by removing $xy$ from~$T$ and adding edge~$uv$ instead.
Now we have that \[ \dep_{T'}(v) = \dep_{T'}(u) + 1 = \dep_T(u) + 1 < \dep_T(v). \]
Thus, $T$ is not canonical.

$\leftarrow$:
Recall that a tree is a join tree of~$H$ if and only if it is a maximum spanning tree of~$G$.
To prove the claim, we construct $T$ via Kruskal's algorithm and show by induction that all constructed subtrees satisfy the claim in each step of the algorithm.
Before that, however, we need a few preliminaries.

We start by partitioning $G$'s edges into equivalence classes such that two edges~$uv, xy$ are in the same class if and only if $uv = xy$.
Afterwards, we sort these classes into a decreasing order~$\langle C_1, C_2, \dots, C_k \rangle$ based on the cardinality of their respective edges.
That is, $C_i$ preceding~$C_j$ implies that $|xy| \geq |uv|$ where $xy \in C_i$ and $uv \in C_j$.
At this point, we can make some interesting observations about Kruskal's algorithm:

\begin{subclaim}\label{claim:fixedEdgeClassOrder}
One can use Kruskal's algorithm to construct all join trees of~\( H \) from any decreasing order of edge classes by only reordering edges within their classes.
\end{subclaim}

\begin{proof}[Claim]
Let $T'$ be an arbitrary join tree of~$H$.
Sort the edges in each class in such a way that edges of~$T'$ precede non-tree edges.
We claim that Kruskal's algorithm constructs~$T'$ from that order.

Assume that the algorithm does not construct~$T'$.
Since the algorithm process all edges and always produces a join tree, that tree must contain an edge~$uv \notin T'$.
Hence, there is an edge~$xy \in T'[u, v]$ which the algorithm did not add, but which connects two separate subtrees at the time $uv$ is processed.
By properties of join trees, $xy \supseteq uv$.
If $xy \supset uv$, then the algorithm processes $xy$ before~$uv$, since it belongs to a class with ``larger'' edges.
Alternatively, if $xy = uv$, then the algorithm also processes $xy$ before~$uv$, because of the way we sorted edges within classes containing them.
In either case, the algorithm adds $xy$ to the resulting tree which contradicts our assumption.
\qedClaim
\end{proof}

\begin{subclaim}\label{claim:sameEdgeSets}
For a given order of edge classes, the partition of nodes into connected components before and after processing a class is an invariant shared by all join trees of~\( H \).
\end{subclaim}

\begin{proof}[Claim]
Assume that there are two different node orders $\sigma$ and~$\tau$ and a node pair~$x,y$ such that, after processing some edge class, $x,y$ belong to the same node set~$S_\sigma$ when processing $\sigma$ but to two different sets when processing~$\tau$.
Let $T_\sigma$ be the subtree formed by~$S_\sigma$ and let $S_{\tau,x}$ be the set containing~$x$ when processing~$\tau$.
Now consider the path~$T_\sigma[x, y]$.
It must contain an edge~$uv$ such that $u \in S_{\tau,x}$ and $v \notin S_{\tau,x}$.
Let $S_{\tau,v}$ be the set containing~$v$.
Since $uv \in T_\sigma$, the algorithm has processed all edges of its equivalent class, \ie, it has processes $uv$ for both node orders.
It would have therefore added $uv$ and combined $S_{\tau,x}$ and $S_{\tau,v}$ into one set, leading to a contradiction.
\qedClaim
\end{proof}

We now perform an induction over $\langle C_1, C_2, \dots, C_k \rangle$.
That is, at step~$i$ of the induction, all edges in~$C_1, \dots, C_i$ have been processed by the algorithm.
For each step, we show that the constructed subtrees of~$T$ are canonical for their respective roots if $T$ does not induce back edges.
Clearly, that statement is correct for $i = 0$ where each subtree only contains a single node.

Assume that all subtrees generated from~$C_1, \dots, C_{i-1}$ satisfy the statement and that we are adding the edges from~$C_i$ next.
Let $\calS_i = \{ S_1, S_2, \dots \}$ denote affected subtrees.
Recall that $T$ is a rooted tree.
Thus, each set~$S_j$ forms a rooted subtree of~$T$.
Additionally, by adding the edges from~$C_i$, the sets in~$\calS_i$ themself form a rooted tree with one of them as root.
Let $R_i$ denote that root.
Let $x_jy_j$ denote an added edge from $y_j \in S_j \neq R$ to its parent~$x_j \notin S_j$.
Hence, $y_j$ is the root node of~$S_j$.

\begin{subclaim}\label{claim:yIsOnlyRoot}
For each subtree~\( S_j \neq R_i \), \( \{ \, u \in S_j \mid u \supseteq x_jy_j \, \} = \{ \, y_j \, \} \). That is, \( y_j \) is the only possible root for~\( S_j \).
\end{subclaim}

\begin{proof}[Claim]
Assume that $S_j$ contains another node~$z_j$ with $z_j \supseteq x_jy_j$.
Hence, $x_jz_j = x_jy_j$ and, by Lemma~\ref{lem:unionJoinProperties}, $x_jz_j \in G$.
However, since only one of $y_j$ and~$z_j$ is the root of~$S_j$ and adjacent to~$x_j$ in~$T$, the other is a non-child descendant of~$x_j$.
Therefore, either $x_jy_j$ or~$x_jz_j$ forms a back edge which contradicts with $T$ not inducing such edges.
\qedClaim
\end{proof}

\begin{subclaim}\label{claim:uniqueX}
For all \( S_j \neq R_i \), \( x_j \) is the highest node in~\( R_i \) such that~\( x_j \supseteq x_jy_j \).
\end{subclaim}

\begin{proof}[Claim]
Firstly, recall that the sets in~$\calS$ form a tree with $R_i$ as root.
Now assume that there is one set~$S_j$ which is child of~$S_k \neq R_i$.
Then, $T$ contains a path $y_j \edge x_j \path y_k \edge x_k$ with $x_j \in S_k$.
However, since $y_j \supseteq x_jy_j = x_ky_k \subseteq x_k$, it follows form Lemma~\ref{lem:unionJoinProperties} that $y_jx_k \in G$ forms a back edge.
Therefore, all $S_j \neq R_i$ are children of~$R_i$.
Lastly, assume that some $x_j$ has an ancestor~$z \in R_i$ with $z \supseteq x_jy_j$.
Then, by Lemma~\ref{lem:unionJoinProperties}, $y_jz \in G$ forms a back edge.
\qedClaim
\end{proof}

We can now complete the inductive step.
Let $u$ be a node in~$\bigcup_{S \in \calS_i} S$ an let $r_i$ be the root node of~$R_i$.
Since Claim~\ref{claim:fixedEdgeClassOrder} and Claim~\ref{claim:sameEdgeSets} imply that creating any join tree would produce the same sets~$\calS_i$ (although the subtrees formed from each set may differ), we prove the statement by showing that the path~$u \path r_i$ in~$T$ has minimal length over all join trees.

The statement is correct by induction hypothesis if $u \in R_i$.
We therefore assume that $u \in S_j \neq R_i$.
In that case, $u \path r_i = u \path y_j \edge x_j \path r_i$.
Note that, by Claim~\ref{claim:yIsOnlyRoot}, $y_j$ is the only possible root for~$S_j$ even when constructing another join tree than~$T$.
Additionally, by induction hypothesis, the distance from $u$ to~$y_j$ is minimal in~$T$.
Clearly, the depth of~$y_j$ is minimal if it is adjacent to the highest~$x \supseteq x_jy_j$.
By Claim~\ref{claim:uniqueX}, $x_j$ is the highest such node and, by induction hypothesis, the distance from $x_j$ to~$r_i$ is minimal in~$T$.
Therefore, $u \path r_i$ has minimal length in~$T$, too.
\qed
\end{proof}

\subsection{Hypergraphs without Canonical Join Trees}

It follows directly from Theorem~\ref{theo:canonIffBackEdge} that a hypergraph does not admit any canonical join tree if and only if each join tree induces a back edge.
In the following, we define a subgraph which enforces that property.
Indeed, we show below (Theorem~\ref{theo:NoCanonicalIffBowtie}) that if one join tree forms it, then all join trees do.

Let $P = w \edge v \edge u \path x \edge y \edge z$ be a path on a join tree~$T$ where $uv \subset vw$ and $xy \subset yz$ and, hence, $uw, xz \in G$ (see Figure~\ref{fig:bowtie}).
We call the subgraph formed by $P$, $uw$, and~$xz$ a \emph{stretched bowtie}.
We define its \emph{length} as the distance in~$T$ from $v$ to~$x$.
In its shortest form (with length~$0$), $x \equiv v$ and $y \equiv u$.
It is equal to a classical bowtie for the case that $u \equiv x$ (length~$1$).
For the rest of this paper, we refer to such a graph simply as \emph{bowtie}, independent of its length.
Note that we do not define bowties as induced graphs.
$G$ may contain additional edges between any nodes of a bowtie.

\begin{figure}[ht]
    \centering
    \includegraphics{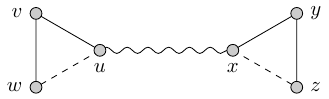}
    \hfil
    \includegraphics{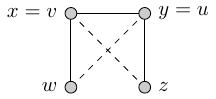}
    \caption
    {%
        A general \emph{bowtie} (left) and one with length~$0$ (right).
        Straight solid lines are tree edges, dashed lines are non-tree edges, and the waved line is a path in the join tree with arbitrary length.
        $G$ may contain additional edges between shown nodes.
    }
    \label{fig:bowtie}
\end{figure}

Theorem~\ref{theo:NoCanonicalIffBowtie} below shows that bowties characterise hypergraphs without any canonical join tree and that a bowtie for one join tree implies a one for all join trees.

\begin{theorem}\label{theo:NoCanonicalIffBowtie}
The following are equivalent.
\begin{enumerate}[$(i)$]
    \item
        \label{item:noCanTree}
        \( H \) does not admit a canonical join tree.
    \item
        \label{item:allTreesHaveBowtie}
        All join trees of~\( H \) form a bowtie.
    \item
        \label{item:someTreeHasbowtie}
        Some join tree of~\( H \) forms a bowtie.
\end{enumerate}
\end{theorem}

\begin{proof}
$\text{\ref{item:noCanTree}} \rightarrow \text{\ref{item:allTreesHaveBowtie}}$:
Let $T$ be an arbitrary join tree for~$H$.
Since it is not canonical, Theorem~\ref{theo:canonIffBackEdge} and Corollary~\ref{cor:2LayerBackEdge} imply that it contains a 2-layer back edge.
Let $u,v,w$ form the lowest such edge with respect to~$T$.
We now construct two additional join trees.
First, let $T_v$ be the resulting join tree from rooting $T$ in~$v$.
We then construct $T_v'$ from~$T_v$ by making each node adjacent to its highest possible ancestor.

Since $T_v'$ is not canonical, it also contains a 2-layer back edge formed by $x,y,z$.
Due to Corollary~\ref{cor:2LayerBackEdge}, we know that $xz = xy \subseteq yz$ or $xy \supseteq yz = xz$.
Additionally, since $y$ is the highest possible ancestor of~$z$, it follows that $xy \nsupseteq yz$ and, therefore, $xz = xy \subset yz$.

Observe that making nodes adjacent to their highest possible ancestor does not introduce new or reverse existing ancestor-descendant relationships.
Therefore, there is a path $x \path y \path z$ in~$T_v$.
Lemma~\ref{lem:unionJoinProperties} then implies that $T_v$ contains two edges $x'y' \in T_v[x, y]$ and $y'z' \in T_v[y, z]$ with $x'y' = xy$ and~$y'z' = yz$.
Additionally, by properties of join trees, $x'y' \subseteq e$ for each edge~$e \in T_v[x', z']$.
Thus, somewhere on that path, there are three consecutive nodes~$a,b,c$ with $ab \subset bc$.
Due to the latter, we can create a new join tree for~$H$ by disconnecting $ab$ and adding edge~$ac$ instead.
Therefore, $ab, ac, bc \in G$ with $ac = ab \subset bc$.

Recall that we have chosen $u, v, w$ as the lowest 2-layer back edge in~$T$.
Because of that and the way we constructed~$T_v$, $x, y, z$ and $a, b, c$ must be in the subtree (of~$T_v$) induced by $v$, $u$ and $u$'s descendants.
Therefore, $w \edge v \edge u \path a \edge b \edge c$ forms a bowtie in~$T_v$ and, since $T$ and~$T_v$ only differ by their roots, also in~$T$.

$\text{\ref{item:someTreeHasbowtie}} \rightarrow \text{\ref{item:noCanTree}}$:
Let $T$ be a join tree for~$H$ containing a bowtie~$\calB$ formed by the nodes $u,v,w$ and $x,y,z$, such that the path~$P = w \edge v \edge u \path x \edge y \edge z$ has minimum length.
That is, there is no join tree~$T'$ for~$H$ where $\dis_{T'}(w, z) < \dis_T(w, z)$.

Assume that $H$ admits a canonical join tree~$T'$, and that we construct~$T'$ by mimicking Kruskal's algorithm.
That is, we start with each node forming their own subtree and step by step connect them with edges in decreasing order.
Consider the state of the algorithm right before adding an edge~$e_a = uv$.
(Note that we do not know if $T'$ contains $uv$, $uw$, or a different edge which corresponds to the same vertices of~$H$.)
Since $uv = uw \subset vw$, there are two disjoint subtrees $T'_{vw}$ and~$T'_u$ which contain the nodes $v,w$ and $u$, respectively.
Adding $e_a$ then combines them into one subtree.
Since $T'$ is rooted (it is canonical), let $a$ be the parent node of~$e_a$.

We first assume that $a \in T'_u$.
Since $a \supseteq e_a = uv = uw$, it follows that $av, aw \in G$.
Also, since $a$ is only adjacent to one node in~$T'_{vw}$, at least one of~$v,w$ is not a child of~$a$ in~$T'$.
Thus, $av$ or~$aw$ forms a back edge in~$G$ with respect to~$T'$.
This contradicts with $T'$ being canonical (see Theorem~\ref{theo:canonIffBackEdge}) and, therefore, implies that $a \notin T'_u$.

Since $a \in T'_{vw}$, $a$ must also be on the ``$v$-side'' of the edge~$uv$ in~$T$.
Additionally, since $a \supseteq e_a = uv = uw$, $au \in G$.
Therefore, $u$ must be a child of~$a$ in~$T'$ (\ie $e_a \equiv au$), because $au$ would otherwise be a back edge.

We can repeat the logic above for $x,y,z$ and conclude that there is a node~$b$ on the ``$y$-side'' of the edge~$xy$ in~$T$ which, in~$T'$, is the parent of~$x$.
Due to their locations in~$T$, we also know that~$a \not\equiv b$.

We now show that $\calB$ must have a length~$\ell > 1$ in~$T$.
First, we consider the case that $\ell = 0$, \ie, $v \equiv x$ and $u \equiv y$.
In such a case, $e_a$ and~$e_b$ are the same edge (since $uv \equiv xy$) and, thus, $a$ and~$b$ must be the same node.
Second, assume that $\ell = 1$, \ie, $u \equiv x$.
Then, again, $a$ and~$b$ must be the same node, since both are the parent of~$u \equiv x$.
Therefore, both cases contradict with the earlier observation that~$a \not\equiv b$.

Note that, due to $\ell > 1$ and their locations in~$T$, the nodes $a$, $b$, $u$ and~$x$ are pairwise distinct.
It therefore follows that, in~$T'$,
\begin{enumerate*}[(1)]
    \item
        \label{case:xNotDescA}
        $x$ is not a descendant of~$a$, or
    \item
        \label{case:uNotDescB}
        $u$ is not as descendant of~$b$
\end{enumerate*}
(or both).
Consider case~\ref{case:xNotDescA}.
Because $\ell > 1$, the length of~$T[u, x]$ is non-zero, too.
Additionally, $T[u, x]$ contains neither~$a$ nor~$b$.
We now follow $T[u, x]$ through~$T'$.
Since $a \notin T[u, x]$ and $x$ is not a descendant of~$a$, there is an edge~$u'x'$ in~$T[u, x]$ such that $u'$ is a descendant of~$a$, $x'$ is not descendant of~$a$, and, thus, $a \in T'[u', x']$.
Note that $a$ and $u'x'$ are on different sides of~$uv$ in~$T$.
Hence, $u'x' \subseteq uv \subset vw$.
It follows that $vx', wx'\in G$ and that $x',v,w$ and $x,y,z$ form a bowtie with strictly smaller length than~$\calB$.
By symmetry, the same applies to case~\ref{case:uNotDescB}:
$u,v,w$ and $u',y,z$ form a bowtie with strictly smaller length.
That contradicts with the minimality of~$\calB$.
Therefore, $T'$ cannot be canonical.
\qed
\end{proof}

\subsection{Hypergraphs with Canonical Join Trees for some Roots}

We can summarize the proof that, in Theorem~\ref{theo:NoCanonicalIffBowtie}, \ref{item:noCanTree} implies~\ref{item:allTreesHaveBowtie} as follows:
A join tree~$T$ for a given hypergraph~$H$ is either canonical, or it has a lowest 2-layer back edge~$u,v,w$.
In case of the latter, if we root $T$ in~$v$ and make each node adjacent to its highest possible ancestor, then the resulting tree is either canonical (\ie, no back edges), or there is no canonical join tree for~$H$.
Based on this, Algorithm~\ref{algo:someCanonicalTree} determines if a given acyclic hypergraph admits a canonical join tree and, in that case, constructs such a tree.

\begin{algorithm}
\caption{%
    Computes a canonical join tree if such a tree exists.
}
\label{algo:someCanonicalTree}

\KwIn{%
    An acyclic hypergraph~$H$.
}

\KwOut{%
    A canonical join tree for~$H$ if such tree exists, or \textsc{Nil} otherwise.
}

\smallskip

Compute join tree~$T$ rooted in arbitrary node~$r$ such that each node is adjacent to its highest possible ancestor (see~\cite{Leitert2021,TarjanYannak1984} or~\cite{LuoVavVanWan2026}).
\label{line:createT}

Find a tripple $u, v, w$ which forms a 2-layer back edge and with maximum~$\dep_T(u)$ (see Algorithm~\ref{algo:All2LayerBackEdge}).
If there is no such tripple, \Return $T$.
\label{line:searchBackEdgeUvw}

\smallskip

Root $T$ in~$v$ and, afterwards, make each node adjacent to its highest possible ancestor.
Let $T_v'$ be the resulting join tree.
\label{line:rootTinV}

Find a tripple $x, y, z$ which forms a 2-layer back edge (see Algorithm~\ref{algo:All2LayerBackEdge}).
If there is no such tripple, \Return $T_v'$.
\label{line:searchBackEdgeXyz}

\smallskip

\Return \textsc{Nil}.
\label{line:returnNil}
\end{algorithm}

\begin{theorem}
Algorithm~\ref{algo:someCanonicalTree} returns a canonical join tree for a given acyclic hypergraph~\( H \) in \( \calO(N) \) time if and only if \( H \) admits such a tree.
\end{theorem}

\begin{proof}
We first consider the case that $H$ does not admit any canonical join tree.
Theorem~\ref{theo:canonIffBackEdge} and Corollary~\ref{cor:2LayerBackEdge} then imply that each join tree induces a 2-layer back edge.
Line~\ref{line:searchBackEdgeUvw} and line~\ref{line:searchBackEdgeXyz}, therefore, both find such an edge.
Subsequently, the algorithm reaches line~\ref{line:returnNil} and returns~\textsc{Nil}.

Next, we consider the case that $T$ (created in line~\ref{line:createT}) is canonical.
Theorem~\ref{theo:canonIffBackEdge} then implies that $T$ does not induce any back edges.
Therefore, line~\ref{line:searchBackEdgeUvw} does not find such an edge and returns~$T$.

Lastly, we consider the case that $H$ admits some canonical join tree, but $T$ (created in line~\ref{line:createT}) is not canonical.
In that case, the algorithm finds a 2-layer back edge~$u,v,w$ in line~\ref{line:searchBackEdgeUvw}, creates a new join tree~$T_v'$ in line~\ref{line:rootTinV}, and checks if $T_v'$ induces a 2-layer back edge~$x,y,z$ in line~\ref{line:searchBackEdgeXyz}.
Assume that line~\ref{line:searchBackEdgeXyz} finds~$x,y,z$.
Recall that nodes in~$T$ and in~$T_v'$ are adjacent to their highest possible parent.
Hence, $uv \subset vw$ and $xy \subset yz$.
Otherwise, $w$ and~$z$ would be adjacent to $u$ and~$x$, respectively.
Additionally, since $u,v,w$ form the lowest 2-layer back edge in~$T$, $x,y,z$ must be in the subtree (of~$T_v'$) induced by $v$, $u$ and $u$'s descendants.
Therefore, $w \edge v \edge u \path x \edge y \edge z$ forms a bowtie in~$T_v'$.
However, Theorem~\ref{theo:NoCanonicalIffBowtie} then implies that no join tree of $H$ is canonical, contradicting with our assumption that such a tree exists.
Therefore, if $H$ admits a canonical join tree, line~\ref{line:searchBackEdgeXyz} does not find a back edge and instead returns~$T_v'$.

To conclude the proof, note that computing a join tree with a given root and making all nodes adjacent to their highest possible ancestor (line~\ref{line:createT} and line~\ref{line:rootTinV}) can be done in $\calO(N)$ time~\cite{Leitert2021,TarjanYannak1984}\cite{LuoVavVanWan2026}.
Additionally, one can determine all 2-layer back edges (line~\ref{line:searchBackEdgeUvw} and line~\ref{line:searchBackEdgeXyz}) in $\calO(N)$ time (see Lemma~\ref{lem:All2LayerBackEdge}).
Therefore, Algorithm~\ref{algo:someCanonicalTree} runs in $\calO(N)$ total time.
\qed
\end{proof}

\subsection{Hypergraphs with Canonical Join Trees for all Roots}

Luo~et\,al.\,\cite{LuoVavVanWan2026} already show that Berge-acyclic hypergraphs admit a canonical join tree for each of their hyperedges as root.
They observe that the line and union join graphs of these hypergraphs are block graphs where each spanning tree is also a join tree and then utilise that property.
Interestingly, Theorem~\ref{theo:allCanonEqui} below shows that their observation almost characterises hypergraphs with canonical join trees for each hyperedge.
However, only the union join graph needs to form a block graph; the line graph does not need to satisfy this property.

\begin{theorem}\label{theo:allCanonEqui}
Let \( H \) be an acyclic hypergraph~\( H \) with a join tree~\( T \), a line graph~\( L \), and a union join graph~\( G \).
Then, the following are equivalent:
\begin{enumerate}[$(i)$]
    \item
        \label{item:allCanonical}
        \( H \) admits a canonical join tree for each of its hyperedges as root.
    \item
        \label{item:noSubsetSeparators}
        For any two edges~\( uv, xy \in T \), \( uv \not\subset xy \).
    % \item
    %     \label{item:sameWeightCuts}
    %     For all join cuts, all edges of both induced subtrees have same weight (i.e. separators contain same vertices).
    % \item
    %     \label{item:unionJoinTriangle}
    %     For any triangle in~\( G \), all its edges have the same weight (\ie, all their respective separators are equal).
    \item
        \label{item:uniformWeightBlockgraph}
        \( G \) is a block graph where all edges within a block have uniform weight.
    \item
        \label{item:linegraphWeights}
        For all \( uv \in L \), either \( |uv| < \min \bigl\{\, |xy| \bigm| xy \in T[u, v] \,\bigr\} \), or \( |uv| = \max \bigl\{\, |xy| \bigm| xy \in T[u, v] \,\bigr\} \).
    % \item
    %     \label{item:canonicalInheritance}
    %     Each acyclic \( \calE' \subseteq \calE \) admits a canonical join tree for each of its hyperedges as root.
    %     TODO: Potentially false? There is no preservation of join trees when removing hyperedges even if acyclicity is preserverd.
\end{enumerate}
\end{theorem}

\begin{proof}
$\text{\ref{item:allCanonical}} \rightarrow \text{\ref{item:noSubsetSeparators}}$:
Let $ab \subset uv$ such that they form a path~$a \edge b \path  u \edge v$ in some join tree~$T$.
Recall that $b \path u$ includes the case where $b \equiv u$.
Since $ab \subset uv$, $au, av \in G$ (disconnect $ab$ and make $a$ adjacent to $u$ or~$v$).
Also, observe that $uv \nsubseteq a$, otherwise $ab$ would need to contain $uv$ since it is on the path from $u$ to~$a$ in~$T$.
Now consider a join tree~$T_a$ rooted in~$a$.
Because $uv \nsubseteq a$, $u$ and~$v$ cannot both be children of~$a$ in~$T_a$.
Therefore, at least one of $au$ and~$av$ forms a back edge and, hence, any such~$T_a$ is not canonical.

$\text{\ref{item:noSubsetSeparators}} \rightarrow \text{\ref{item:uniformWeightBlockgraph}}$:
Before proving the statement, we make an auxiliary observation.
Let $uv \in G$ be a non-tree edge with respect to some join tree~$T$.
By properties of join trees, $uv \subseteq xy$ for each edge~$xy \in T[u, v]$.
Thus, by property~\ref{item:noSubsetSeparators}, $uv = xy$ for all such edges~$xy$.

Let $C$ be a cycle in~$G$.
Since it is a cycle, $C$ contains at least one edge~$uv \notin T$.
We first consider the case that there are at least two such edges.
In such case, $T[u, v]$ contains an edge~$xy$ which is not part of~$C$.
Since $xy = uv$ (by the above observation), we can construct a slightly modified join tree from~$T$ by removing $xy$ and using $uv$ instead.

We repeat the previous step until only one edge of~$C$ belongs to the resulting join tree.
Again, let $uv$ be that edge.
It follows that $C$ contains exactly $uv$ and~$T[u, v]$, that all edges of~$C$ are equal to~$uv$, and that $x \supseteq uv$ for each node~$x \in C$.
Therefore, for each pair $x,y \in C$, we can construct a join tree containing the edge~$xy$ with $xy = uv$ by removing any edge on~$T[x, y]$ and replacing it with~$xy$ instead.

It follows from the above that each biconnected component of~$G$ forms a clique, making~$G$ a block graph, and that the edges in each block have uniform weight.

$\text{\ref{item:uniformWeightBlockgraph}} \rightarrow \text{\ref{item:allCanonical}}$:
To prove the statement, we effectively repeat the argument from~\cite{LuoVavVanWan2026}.
Pick an arbitrary node and construct a BFS tree~$T$ from it.
Since edge weights within a block are uniform, $T$ is a maximum spanning tree of~$G$ and, thus, a valid join tree for~$H$.
Additionally, since all blocks are cliques and $T$ is a BFS tree, all non-tree edges of~$G$ are cross-edges, making $T$ canonical.

$\text{\ref{item:uniformWeightBlockgraph}} \rightarrow \text{\ref{item:linegraphWeights}}$:
First, consider an edge $uv \in G \cap T$, \ie, $uv$ is a tree edge of~$G$ with respect to~$T$.
Clearly, $|uv| = \max \bigl\{\, |xy| \bigm| xy \in T[u, v] \,\bigr\}$ since it is the only edge on that path.
Next, assume that $uv \in G \setminus T$, \ie, $uv$ is a non-tree edge.
Since $G$ is a block graph, each edge~$xy \in T[u, v]$ is part of the same block and, by property~\ref{item:uniformWeightBlockgraph}, has the same weight as~$uv$.
Thus, $|uv| = \max \bigl\{\, |xy| \bigm| xy \in T[u, v] \,\bigr\}$.

Lastly, consider an edge~$uv \in L \setminus G$ and assume that there is a join tree $T$ for which $|uv| \geq \min \bigl\{\, |xy| \bigm| xy \in T[u, v] \,\bigr\}$.
Then, we can create a tree~$T'$ by removing the minimum-weight edge~$xy \in T[u, v]$ from~$T$ and replacing it with~$uv$.
Recall that $T$ is a maximum spanning tree of $L$ and~$G$.
Since $|uv| \geq |xy|$, either $T$ was not a maximum spanning tree (\ie, not a valid join tree), or $T$ and $T'$ are both maximum spanning trees and, hence, $uv \in G$.
Either case contradicts with our assumptions.
Therefore, for all such $uv$, $|uv| < \min \bigl\{\, |xy| \bigm| xy \in T[u, v] \,\bigr\}$.

$\text{\ref{item:uniformWeightBlockgraph}} \leftarrow \text{\ref{item:linegraphWeights}}$:
We start by considering the edges~$uv \in L$ with $|uv| < \min \bigl\{\, |xy| \bigm| xy \in T[u, v] \,\bigr\}$.
Since $|uv|$ is strictly less than the minimum, $uv$ is not an edge of~$T$.
Assume now that $L$ has a maximum spanning tree~$T'$ which contains~$uv$.
Let $T'_u$ and~$T'_v$ be the respective subtrees of~$T'$ after removing~$uv$.
Since $uv \notin T$, $uv$ and $T[u, v]$ form a cycle and, thus, there is an edge~$xy \in T[u, v]$ with $x \in T'_u$ and $y \in T'_v$.
By construction, $|xy| > |uv|$.
Therefore, by connecting $T'_u$ and~$T'_v$ via $xy$, we create a spanning tree with greater total weight than~$T'$.
In other words, there is no maximum spanning tree for~$L$ which contains $uv$.

It follows from the above that the edges of~$G$ are exactly the edges~$uv$ with $|uv| = \max \bigl\{\, |xy| \bigm| xy \in T[u, v] \,\bigr\}$ for some maximum spanning tree~$T$.
We now show that each cycle of~$G$ forms a clique by repeating the approach from~$\text{\ref{item:noSubsetSeparators}} \rightarrow \text{\ref{item:uniformWeightBlockgraph}}$:

First, observe that $uv = xy$ for all $xy \in T[u, v]$ if $uv$ is a non-tree edge, because $xy \subseteq uv$ and $|uv| \geq |xy|$.
Then, assume that $G$ contains a cycle~$C$ and use the previous observation to modify $T$ until the resulting join tree contains all but one edge of~$C$.
It follows from the same observation that, for each pair $x,y \in C$, we can construct a join tree containing the edge~$xy$ with $xy = uv$.
Therefore, $G$ is a block graph and the edges in each block have uniform weight.
\qed
\end{proof}

We conclude the paper with an observation on how hard it is to determine if a given hypergraph has a canonical join tree for each of its hyperedges as root.
The question if a given family of sets~$\calF$ contains two sets $S, S'$ such that $S \subset S'$ is known as \emph{Sperner Family} problem.
If the Strong Exponential Time Hypothesis is true, then there is no algorithm which solves the the Sperner Family problem in $\calO \bigl( N^{2-\varepsilon} \bigr)$ time for some constant~$\varepsilon > 0$~\cite{BoraCresHabi2016}.

Based on property~\ref{item:noSubsetSeparators} of Theorem~\ref{theo:allCanonEqui}, we can now build a simple reduction.
Let $\calF = \{ S_1, S_2, \dots, S_n \}$ be an arbitrary family of sets with $N = \sum_{S_i \in \calF} |S_i|$.
We create an acyclic hypergraph~$H$ as follows:
Create a hyperedge~$U$ which is the union of all sets~$S_i \in \calF$ and create a hyperedge $E_i = \{ v_i \} \cup S_i$ for each $S_i \in \calF$ where $v_i$ is a new vertex that is only in~$E_i$.
Now, let $\calE = \{ U, E_1, \dots, E_n \}$ be the hyperedges of~$H$.
We can create a join tree~$T$ for $H$ by making each $E_i$ adjacent to~$U$.
Note that $E_i \cap U = S_i$.
Hence, the edge set of~$T$ is equal to~$\calF$ and, therefore, $H$ has a canonical join tree for each of its hyperedges as root if and only if $\calF$ does not contain two sets $S_i, S_j$ with $S_i \subset S_j$.
We can therefore conclude:

\begin{corollary}
If the Strong Exponential Time Hypothesis is true, then there is no algorithm which, for a constant~$\varepsilon > 0$, decides in $\calO \bigl( N^{2 - \varepsilon} \bigr)$ time whether or not a given acyclic hypergraph has a canonical join tree for each of its hyperedges as root.
\end{corollary}

\end{document}